\documentclass[acmsmall,nonacm]{acmart}

\usepackage{enumitem}
\usepackage{xspace}
\usepackage{mathtools} 
\usepackage{relsize} 
\usepackage{adjustbox} 
\usepackage{xifthen}
\usepackage{multirow}
\usepackage{array}
\usepackage{mathdots} 
\usepackage[normalem]{ulem} 
\usepackage{forest}
\usetikzlibrary{backgrounds}
\usetikzlibrary{snakes}
\usetikzlibrary{arrows.meta}
\usepackage{thmtools}
\declaretheorem{property}
\declaretheorem{observation}
\declaretheorem{recurrence}

\newcommand{\NAME}[1]{{\operatorname{\sf #1}}}

\makeatletter
\newcommand{\FUNCTIONx}[1]{%
  \@ifnextchar(
  {\ensuremath{{#1}}}%
  {\FNx{#1}}%
}
\newcommand{\FUNCTIONxx}[1]{%
  \@ifnextchar(
  {\ensuremath{{#1}}}%
  {\FNxx{#1}}%
}
\makeatother
\newcommand{\FNx}[2]{\ensuremath{{#1}({#2})}} 
\newcommand{\FNxx}[3]{\ensuremath{{#1}({#2}, {#3})}}

\newcommand{\NAMEDFNx}[2]{\FUNCTIONx{\NAME{#1}}{#2}}

\newcommand{\NAMEDFNsubx}[3]{\ensuremath{\NAME{#1}_{#2}({#3})}}

\newcommand{\heavier}{\NAMEDFNx{heavier}}
\newcommand{\lighter}{\NAMEDFNx{lighter}}
\newcommand{\holes}{\NAMEDFNx{holes}}
\newcommand{\lightholes}{{H}}

\newcommand{\probName}[1]{{\mbox{{#1}}}\xspace}
\newcommand{\threeWCST}{\probName{{\small\sf 3WST}}}
\newcommand{\twoWCST}{\probName{{\small\sf 2WST}}}
\newcommand{\twoWCDT}{\probName{{\small\sf 2WDT}}}
\newcommand{\LDT}{\probName{{\small\sf LDT}}}

\newcommand{\threeWCSTs}{\threeWCST{}s\xspace}
\newcommand{\twoWCSTs}{\twoWCST{}s\xspace}
\newcommand{\twoWCDTs}{\twoWCDT{}s\xspace}

\newcommand{\LDTs}{\LDT{}s\xspace}

\newcommand{\twoWCDTsintitle}{\probName{{\sf 2WDTs}}}

\newcommand{\heaviest}{\NAMEDFNsubx{heaviest}}  

\newcommand{\setyes}[2]{{#1}^{\scriptstyle \NAME{yes}}_{\scriptstyle #2}}
\newcommand{\setno}[2]{{#1}^{\scriptstyle \NAME{no}}_{\scriptstyle #2}}

\newcommand{\weightName}{w}
\newcommand{\bweightName}{\tilde \weightName}
\newcommand{\pweightName}{\weightName^*}

\newcommand{\weight}[1]{{\weightName({#1})}}
\newcommand{\bweight}[1]{{\bweightName({#1})}}
\newcommand{\pweight}[1]{{\pweightName({#1})}}

\newcommand{\cost}{\NAMEDFNx{cost}}
\newcommand{\pcost}{\FUNCTIONx{\NAME{cost}^*}}
\newcommand{\Dictionary}{{D}}
\newcommand{\AdmSubsets}{\mathcal{A}}
\newcommand{\acost}[1]{\NAMEDFNsubx{cost}{\AdmSubsets} {#1}}
\newcommand{\Matrix}{{M}}

\newcommand{\test}[2]{\ensuremath{\langle {} \mathrel{#1} #2 \rangle}}

\newcommand{\undefsymbol}{\bot}

\newcommand{\calC}{\mathcal C}
\newcommand{\calF}{\mathcal F}

\newcommand{\calH}{\mathcal H}
\newcommand{\calL}{\mathcal L}

\newcommand{\calR}{\mathcal R}
\newcommand{\calS}{\mathcal S}

\newcommand{\braced}[1]{{ \left\{ #1 \right\} }}

\newcommand{\uNAME}{u}
\newcommand{\uOF}[1]{\uNAME_{#1}}
\newcommand{\uA}{\uOF 1}
\newcommand{\uB}{\uOF 2}
\newcommand{\uC}{\uOF 3}
\newcommand{\uF}{\uOF 4}
\newcommand{\uFp}{\uOF 5}
\newcommand{\uFpp}{\uOF 6}
\newcommand{\ui}{\uOF i}
\newcommand{\uJ}{\uOF j}

\newcommand{\uD}{\uOF d}

\newcommand{\up}{\uOF p}
\newcommand{\upp}{\uOF {p+1}}
\newcommand{\upm}{\uOF {p-1}}
\newcommand{\uq}{\uOF q}
\newcommand{\uqp}{\uOF {q+1}}
\newcommand{\uqm}{\uOF {q-1}}
\newcommand{\uip}{\uOF {i+1}}
\newcommand{\uJp}{\uOF {j+1}}

\newcommand{\uDm}{\uOF {d-1}}

\newcommand{\uAto}{\uA \to \uB}
\newcommand{\uBto}{\uB \to \uC}

\newcommand{\uito}{\ui \to \uip}

\newcommand{\etal}{et al.\ }

\newcommand{\qint}[2]{\left[ #1, #2 \right]_{\scriptscriptstyle Q}}
\newcommand{\kint}[2]{\left[ #1, #2 \right]_{\scriptscriptstyle K}}
\newcommand{\setleft}[1]{{\min {#1}}}
\newcommand{\setright}[1]{{\max {#1}}}

\newcommand{\setkint}[1]{\kint {\setleft {#1}} {\setright {#1}}}

\newcommand{\mycase}[1]{\smallskip\par\noindent\normalfont{\underline{\emph{#1}}.}}

\newcommand{\treeFontSize}{\small}

\newcommand{\interiornodeFill}{black!4}
\newcommand{\leafnodeFill}{black!4}
\newcommand{\subtreeFill}{black!4}
\newcommand{\textColor}{black}
\newcommand{\drawColor}{black!50}
\newcommand{\edgeColor}{black!60}
\newcommand{\edgeThickness}{}
\newcommand{\drawThickness}{thin}

\newcommand{\dimFill}{black!2}
\newcommand{\dimText}{black!65}
\newcommand{\dimDraw}{black!20}
\newcommand{\dimEdge}{black!20}
\newcommand{\dimEdgeThickness}{thin}

\newcommand{\halfdimFill}{black!2}
\newcommand{\halfdimText}{\textColor}
\newcommand{\halfdimDraw}{black!25}
\newcommand{\halfdimEdge}{black!30}
\newcommand{\halfdimEdgeThickness}{}

\newcommand{\brightFill}{red!4}
\newcommand{\brightText}{\textColor}
\newcommand{\brightDraw}{\drawColor}
\newcommand{\brightEdge}{\edgeColor}

\newcommand{\otherbrightFill}{green!4}
\newcommand{\otherbrightText}{\textColor}
\newcommand{\otherbrightDraw}{\drawColor}
\newcommand{\otherbrightEdge}{\edgeColor}

\newcommand{\afterFill}{blue!4}
\newcommand{\afterText}{\textColor}
\newcommand{\afterDraw}{\drawColor}
\newcommand{\afterEdge}{\edgeColor}

\renewcommand{\interiornodeFill}{black!0}
\renewcommand{\leafnodeFill}{black!0}
\renewcommand{\subtreeFill}{black!0}
\renewcommand{\dimFill}{black!0}
\renewcommand{\halfdimFill}{black!0}

\tikzset{label distance=-1pt}

\tikzset{weight/.style = {label={270:{\scriptsize #1}}}} 
\tikzset{decoration={snake,amplitude=.5mm,segment length=1.5mm,post length=0.4mm,pre length=0.75mm}}

\forestset{
  subtreeStyling/.style={
    draw={\drawColor, \drawThickness},
    edge={\edgeColor, \edgeThickness},
    fill=\subtreeFill,
    text=\textColor,
  },
  dim/.style={
    draw={\dimDraw, \drawThickness},
    edge={\dimEdge, \dimEdgeThickness},
    fill=\dimFill,
    text=\dimText,
  },
  halfdim/.style={
    draw={\halfdimDraw, \drawThickness},
    edge={\halfdimEdge, \halfdimEdgeThickness},
    fill=\halfdimFill,
    text=\halfdimText,
  },
  bright/.style={
    draw={\brightDraw, \drawThickness},
    edge={\brightEdge, \edgeThickness},
    fill=\brightFill,
    text=\brightText,
  },
  otherbright/.style={
    draw={\otherbrightDraw, \drawThickness},
    edge={\otherbrightEdge, \edgeThickness},
    fill=\otherbrightFill,
    text=\otherbrightText,
  },
  afterbright/.style={
    draw={\afterDraw, \drawThickness},
    edge={\afterEdge, \edgeThickness},
    fill=\afterFill,
    text=\afterText,
  },
  interiornodeStyling/.style = {
    draw={\drawColor, \drawThickness},
    edge={\edgeColor, \edgeThickness},
    fill=\interiornodeFill,
  },
  leafnodeStyling/.style = {
    draw={\drawColor, \drawThickness},
    edge={\edgeColor, \edgeThickness},
    fill=\leafnodeFill,
  },
  darkedge/.style={
    edge={black!80, line width=0.8pt, ->},
  },
  dimedge/.style={
    edge={\dimEdge},
  },
  normedge/.style={
    edge={\edgeColor},
  },
  squigglyedge/.style={
    edge={decorate},
  },
  dashededge/.style={
    edge={dotted, line width=1.5pt},
  },
  nodraw/.style = {draw=none, fill=none, anchor=north},
  every leaf node/.style={
    if n children=0{#1}{}
  },
  every interior node/.style={
    if n children=0{}{#1}
  },
  subtreeBase/.style={
    anchor=north,
    child anchor=north,
    outer sep=0pt,
    rounded corners=0.3em,
    shape=semicircle,
    inner sep=1pt,
    text depth=4pt,
    xscale=0.91,
    content format={
      \noexpand\scalebox{1.1}[1]{\forestoption{content}}
    },
  },
  subtree/.style={
    subtreeBase,
    subtreeStyling,
  },
  dimtree/.style={
    subtree,
    dim,
  },
  brighttree/.style={
    subtree,
    bright
  },
  interiornode/.style = {
    interiornodeStyling,
    rounded rectangle,
    inner sep=1.8pt,
    text height=1.6ex,
    text depth=.5ex,
    minimum width=25.5pt, 
    anchor=north,
  },
  leafnode/.style = {
    leafnodeStyling,
    rectangle,
    minimum width=13.5pt,
    inner sep=1.8pt,
    text height=1.7ex,
    text depth=.3ex,
    child anchor=north,
    anchor=north,
  },
}

\tikzstyle{edgeCommon} = [circle, font=\scriptsize, inner sep=0pt, outer sep=1.5pt, pos=0.55]
\tikzstyle{edgeYes} = [edgeCommon, text depth=0ex, text height=1.2ex, auto=right, node contents={y}]
\tikzstyle{edgeNo} = [edgeCommon, text depth=0ex, text height=1.2ex, auto=left, node contents={n}]
\tikzstyle{edgeYesRight} = [edgeCommon, text depth=0ex, text height=1.2ex, auto=left, node contents={y}]
\tikzstyle{edgeNoLeft} = [edgeCommon, text depth=0ex, text height=1.2ex, auto=right, node contents={n}]
\tikzstyle{edgeLess} = [edgeCommon, auto=right, node contents={$<$}]
\tikzstyle{edgeGreater} = [edgeCommon, auto=left, node contents={$>$}, pos=0.5]
\tikzstyle{edgeEqual} = [edgeCommon, auto=right, node contents={$=$}, pos=0.4]

\newenvironment{generictrees}{
  \treeFontSize
  \forest
  where level={0}{}{
    where level={1}{}{
      every leaf node={leafnode},
      every interior node={interiornode},
    }
  },
    }{\endforest}

\newcommand{\inputTree}[1]{%
  \input{FIGURES/TREES/#1}
}

\newcommand{\A}{(a)\xspace}     
\newcommand{\B}{(b)\xspace}
\newcommand{\C}{(c)\xspace}

\begin{document}

\title{Classification via Two-Way Comparisons}
\titlenote{An extended abstract of this paper appears in WADS 2023~\cite{DBLP:conf/wads/ChrobakY23}.
The journal version appears in~\cite{ChrobakY24}.}

\author{Marek Chrobak}
\authornote{Research partially supported by National Science Foundation grant CCF-2153723.}
\orcid{0000-0002-8673-2709}
\affiliation{%
    \institution{University of California, Riverside}
    \city{Riverside}
    \state{California}
    \country{USA}
}

\author{Neal E. Young}
\orcid{0000-0001-8144-3345}
\affiliation{%
    \institution{University of California, Riverside}
    \city{Riverside}
    \state{California}
    \country{USA}
}

\authorsaddresses{}


\begin{abstract}
  Given a weighted, ordered query set $Q$ and a partition of $Q$ into classes,
  we study the problem of computing a minimum-cost decision tree
  that, given any query $q\in Q$, uses equality tests and less-than tests
  to determine $q$'s class.
  Such a tree can be
  faster and smaller than a conventional search tree
  and
  smaller than a lookup table
  (both of which must identify $q$, not just its class).
  We give the first polynomial-time algorithm for the problem.
  The algorithm extends naturally to the setting where each query has multiple allowed classes.
\end{abstract}




\maketitle


\section{Introduction}\label{sec: introduction}

Given a weighted, ordered \emph{query} set $Q$ partitioned into classes,
we study the problem of computing a decision tree
that, given any query $q\in Q$,
uses equality tests (e.g., ``$q=4{?}$'') and less-than tests (e.g., ``$q<7{?}$'')
to quickly determine $q$'s class.
We call such a tree a \emph{two-way-comparison decision tree (\twoWCDT)}.
Figure~\ref{fig: lead example} shows an example.
In the special case where each class is a singleton (so identifies the query),
we call such a tree a \emph{two-way-comparison search tree (\twoWCST)}.
The goal is to find a \twoWCDT of minimum \emph{cost},
defined as the weighted sum of the depths of all queries,
where the depth of a given query $q\in Q$ is the number of tests the tree makes when
processing query $q$.

\begin{figure}[t]
  \centering
  \inputTree{lead-example}
  \caption{
    An optimal two-way-comparison decision tree (\twoWCDT) for the problem instance shown
    on the right. The instance (but not the tree)
    is from~\cite[Figure~6]{chamber_chen_dispatching_1999,Chambers:1999:EMP:320385.320407}.
    Each internal node represents a comparison between the given query and the node's key $k$: either an
    equality test, represented as ``$=\!k$'', or a less-than test, represented as ``$<\!k$''.
    Each leaf (rectangle) is labeled with the queries that reach it, and below that with the class for the leaf.
    The table gives the class and weight of each query $q\in Q= [50] = \{1, 2,\ldots, 50\}$.
    The tree has cost 2055, about 11\% cheaper
    than the tree from~\cite{chamber_chen_dispatching_1999,Chambers:1999:EMP:320385.320407},
    of cost 2305.
  }\label{fig: lead example}
\end{figure}

Whereas search trees and lookup tables must identify the query $q$
(or the inter-key interval that $q$ lies in),
a decision tree needs only to identify $q$'s class,
so can be faster and smaller than a conventional search tree,
and smaller than a lookup table.
Consequently,
decision trees are used in applications such as dispatch trees,
which allow compilers and interpreters to quickly resolve
method implementations for objects declared with type inheritance~\cite
{chamber_chen_dispatching_1999,Chambers:1999:EMP:320385.320407}.
(Each type is assigned a numeric ID via a depth-first search of the inheritance digraph.
For each method, its tree maps each type ID to its method resolution.)
Chambers and Chen~\cite
{chamber_chen_dispatching_1999,Chambers:1999:EMP:320385.320407}
give a heuristic to construct low-cost \twoWCDTs,
but leave open whether minimum-cost \twoWCDTs can be found in polynomial time.

We give the first polynomial-time algorithm to find minimum-cost \twoWCDTs.
It runs in time $O(n^4)$,
where $n=|Q|$ is the number of distinct query values.
This matches the best run-time known for the special case of \twoWCSTs.
The algorithm extends naturally to the setting where each query can belong to multiple classes,
any one of which is acceptable as an answer for the query.
The extended algorithm runs in time $O(n^3 m)$, where
$m$ is the sum of the sizes of the classes.


\paragraph{Related work.}
Decision trees of various kinds are ubiquitous in the areas of artificial intelligence, machine learning,
and data mining, where they are used for data classification, clustering, and regression
(see e.g.~\cite{bertsimas_optimal_2017}).
Here we study decision trees for one-dimensional data.
Most work on such trees has focussed on search trees.
Here is a summary of relevant work on optimal search trees.

The tractability of finding an optimal search tree depends heavily
on the kind of tests that the tree may use.
The most general case, allowing tests of membership in sets from any given family of subsets of $Q$,
is NP-hard, even if all subsets have size at most three~\cite{hyafil_constructing_1976},
or the family is required to be laminar~\cite{jacobs_complexity_2010}.
Early works considered trees in which each test compared the given query value $q$
to some particular comparison key $k$, with \emph{three} possible outcomes:
the query value $q$ is less than, equal to, or greater than $k$~\cite
[\S 14.5]{cormen_introduction_2022}~\cite[\S 6.2.2]{Knuth1998}.
We call such a tree a \emph{three-way-comparison} search tree, or \threeWCST for short.
(See Figure~\ref{fig: three-way vs two-way}\,(a).)
In a \threeWCST, the query values that reach any given node form an interval.
The possible intervals naturally represent $O(n^2)$ dynamic-programming subproblems,
leading to an $O(n^3)$-time algorithm for finding minimum-cost \threeWCSTs~\cite{gilbert_variablelength_1959}.
Knuth reduced the running time to $O(n^2)$~\cite{Knuth1971}.


\begin{figure}[t]
  \centering
  \inputTree{three-way-vs-two-way}
  \caption{
    Tree \emph{(a)} is a three-way-comparison search tree (\threeWCST).
    Tree \emph{(b)} is a two-way-comparison search tree (\twoWCST) for the same instance.
    The query (or interval of queries) reaching each (rectangular) leaf is within the leaf.
    The weight of the query (or interval) is below the leaf.
  }\label{fig: three-way vs two-way}
\end{figure}


In practice each three-way comparison is sometimes implemented by doing two two-way tests:
a less-than test followed by an equality test.
Knuth~\cite[\S 6.2.2, Example 33]{Knuth1998} proposed
exploring binary search trees that use these two types of tests directly in any combination,
that is, \twoWCSTs as defined earlier.
For the so-called \emph{successful-queries} variant (defined later),
assuming the query weights are normalized to sum to 1,
there is always a \twoWCST whose cost exceeds the entropy of the weight distribution by at most 1~\cite{dagan_etal_twenty_questions_2017}.
As the entropy is a lower bound on the cost of any binary search tree using arbitrary Boolean tests,
this suggests that restricting to less-than and equality tests need not be too costly.

Stand-alone equality tests introduce an algorithmic obstacle not encountered with \threeWCSTs.
Namely, while (analogously to \threeWCSTs) each node of a \twoWCST is naturally associated with an
interval of queries, not all queries from this interval necessarily reach the node,
so the dynamic program for \threeWCSTs does not extend easily to \twoWCSTs.
This led early works to focus on restricted classes of \twoWCSTs, namely
\emph{median split trees}~\cite{Sheil1978} and \emph{binary split trees}~\cite{Huang1984,Perl1984,Hester1986}.
These, by definition, constrain the use of equality tests
so as to sidestep the obstacle they introduce.
\emph{Generalized binary split trees} are less restrictive,
but the only algorithm proposed to find them~\cite{StephenHuang1984} is incorrect~\cite{chrobak_huang_2022}.
Likewise, the recurrence relations underlying the first algorithms proposed to find minimum-cost \twoWCSTs
(which were given without proof~\cite{Spuler1994Paper,Spuler1994Thesis})
are demonstrably wrong~\cite{chrobak_huang_2022}.

Spuler conjectured in 1994 that every \twoWCST instance has an optimal tree with the \emph{heaviest-first} property:
namely, in each equality-test node,
\emph{the comparison key is the heaviest among keys that reach the node}~\cite{Spuler1994Thesis}.
In 2002 Anderson \etal proved the conjecture for successful-queries \twoWCSTs,
leading to the first polynomial-time algorithm for that variant~\cite{Anderson2002}.
The algorithm runs in $O(n^4)$ time.
In 2021, Chrobak \etal simplified their result (in particular, the handling of equal-weight keys, as discussed later)
to obtain an $O(n^4)$-time algorithm to find optimal \twoWCSTs (both variants)~\cite{chrobak_simple_2021}.
%
%
These \twoWCST algorithms do not extend easily to \twoWCDTs,
because some \twoWCDT instances have no optimal tree with the heaviest-first property.
Figure~\ref{fig: inversion} gives an example.

\begin{figure}[t]
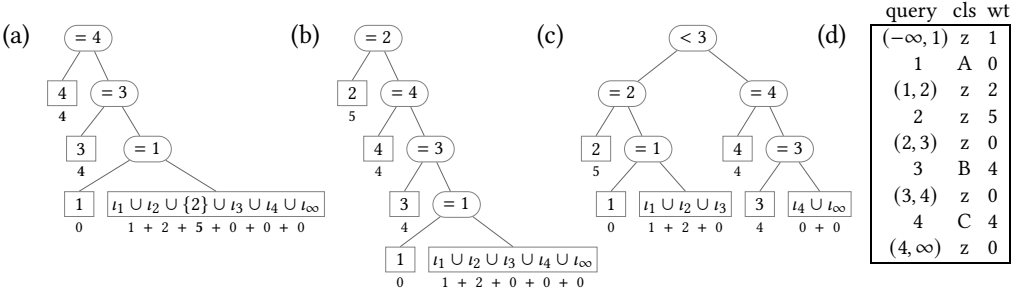

  \centering
  \inputTree{inversion}
  \caption{\label{fig: inversion}%
  Three trees for the \twoWCDT instance shown in \emph{(d)}.
  The set of queries reaching each (rectangular) leaf is shown within the leaf
    (to save space, there $\iota_i$ denotes the inter-key open interval
    with right boundary $i$, e.g. $\iota_1 = (-\infty, 1)$, $\iota_2 = (1, 2)$).
    The associated weights are below the leaf.
    The optimal tree \emph{(a)} has cost 36 and is not heaviest-first.
    Each heaviest-first tree (e.g.\ \emph{(b)} of cost 37
    or \emph{(c)} of cost 39) is not optimal.
    These properties also hold if each weight is perturbed 
    to make the weights distinct.
    (Note: in our formal model, the inter-key intervals will be represented by virtual non-key queries.)
  }
\end{figure}


\paragraph{Our contributions.}
The \emph{rotation} operation is a standard tool for studying \twoWCSTs with only less-than tests (and \threeWCSTs).
Following~\cite{Anderson2002} and~\cite{chrobak_simple_2021}
we use a generalized rotation that applies to \twoWCSTs with both types of tests.
We generalize it further, to decision trees $T$ such that the test at each internal node $u$
  is a test of membership in some set $X_u\subseteq Q$,
  subject only to the constraint that the collection of such test sets
  $\{X_u : u \in T\}$ is laminar.
  For each such node $u$, the edge to one child is associated with $X_u$
  while the edge to the other child is associated with the complement $\overline X_u = Q\setminus X_u$.
  Given any query $q$, the search for $q$ in $T$ follows the unique root-to-leaf path whose edges' sets all contain $q$.
We call such trees \emph{laminar decision trees}, or {\LDT}s for short.
(See Section~\ref{sec: definitions}.)

Suppose that, in such a laminar decision tree $T$,
there is an ``imbalance'' in the tree:
for some downward path $\uAto\to\cdots\to\uD$,
the sibling $\uB'$ of $\uB$
is lighter than $\uD$.
(That is, $\weight{\uB'} < \weight{\uD}$,
where $\weight u$, the \emph{weight} of node $u$,
is as usual the total weight of the queries that reach the node.)
Then Theorem~\ref{thm: imbalance} (Section~\ref{sec: imbalance}) states that,
if $T$ is optimal,
\emph{the sets associated with the edges leaving the path $\uA\to\cdots\to\uD$ must be pairwise disjoint.}
(The edges leaving the path are $\{\uito': 1\le i < d\}$,
  where, for any node $u$ other than the root, $u'$ denotes the sibling of $u$ in $T$.)
This theorem generalizes the key structural theorems of~\cite{Anderson2002} and~\cite{chrobak_simple_2021};
in particular, it implies the heaviest-first property for \twoWCSTs.

Section~\ref{sec: intermediate} then proves Theorem~\ref{thm: intermediate},
  which strengthens Theorem~\ref{thm: imbalance} specifically for trees with less-than and equality tests,
that is, \twoWCSTs.
  Section~\ref{sec: admissible}
  uses Theorem~\ref{thm: intermediate}
  to prove Theorem~\ref{thm: admissible},
  that there is always an optimal tree that is \emph{admissible}.
  This means roughly that, at each equality-test node $\test = h$ in the tree,
  if the key $h$ is not the heaviest key reaching the node,
  it must be one of at most three other suitably restricted keys
  (Definition~\ref{def: admissible}).
A careful implementation
then yields the main result
  (Theorem~\ref{thm: algorithm} in Section~\ref{sec: algorithm}):
an $O(n^3 m)$-time dynamic-programming algorithm to find a minimum-cost \twoWCDT.


\paragraph{The role of distinct key weights.}
The discussion above glosses over a secondary technical obstacle for \twoWCSTs.
For \twoWCST instances whose key weights are \emph{distinct},
the heaviest-first property determines the key of each equality test uniquely,
so that the queries that reach any given node in a \twoWCST (with the property)
must form one of $O(n^4)$ predetermined subsets,
leading naturally to a dynamic program with $O(n^4)$ subproblems. 
But this uniqueness is lost when key weights are not distinct.
This obstacle turns out to be more challenging than one might expect.
Indeed, there are instances with non-distinct weights for which, for every non-empty subset $S$ of $Q$,
there is a \twoWCST that has the heaviest-first property,
and a node $u$ such that the set of queries reaching $u$ is $S$.
One cannot just break ties naively:
it can be that, for two maximum-weight keys $h$ and $h'$ reaching a given node $u$,
there is an optimal subtree in which $u$ does an equality-test to $h$,
but none in which $u$ does an equality-test to $h'$~\cite[Figure 3]{chrobak_simple_2021}.
Similar issues arise in finding optimal \emph{binary split trees}---these
can be found in time $O(n^4)$ if the instance has distinct weights,
while for arbitrary instances the best bound known is $O(n^5)$~\cite{Hester1986}.

Nonetheless, using a perturbation argument Chrobak~\etal~\cite{chrobak_simple_2021} show that
an arbitrary \twoWCST instance can indeed be handled as if it is a distinct-weights instance
just by breaking ties among equal weights in
a globally consistent way.
We use the same approach here for \twoWCDTs.


\subsection{Definitions}\label{sec: definitions}


\medskip

{ \addtolength{\parskip}{0.2em}
  An instance $I$ of the \emph{laminar decision tree} problem ({\LDT}\!) is specified by a tuple $I=(Q, w, \calC, \calF)$,
  where $Q$ is a finite, totally ordered, non-empty set of \emph{queries}, 
 with each query $q\in Q$ assigned a weight $w(q)\ge 0$,
  the set $\calC\subseteq 2^Q$ is a collection of query \emph{classes}
  (with each class having a unique identifier),
  and $\calF \subseteq 2^Q \setminus \{\emptyset, Q\}$ is laminar.
  Call each set $X\in\calF$ a \emph{test},
  with two \emph{outcomes}: $X$ (the \emph{yes} outcome),
  and $\overline X = Q\setminus X$ (the \emph{no} outcome).
  Let $n$ and $m$ denote, respectively, $|Q|$ and $\sum_{c\in \calC} |c|$.
  A \emph{decision tree} for $I$
  is a rooted binary tree $T$ where each non-leaf node $u$
  has an associated test $X_u\in \calF$,
  with the edge to one child of $u$ associated with the yes-outcome $X_u$,
  and the edge to the other child of $u$ associated with the no-outcome $\overline X_u$.
  Each leaf node $u$ is labeled with (the identifier of) some class $c_u\in\calC$,
  which must contain the intersection of the outcomes of the edges along the path from the root to $u$
  (this intersection is comprised of those queries $q\in Q$ whose search, as defined next, ends at $u$).

  For each $q\in Q$, the \emph{search for $q$ in $T$} follows
  the (unique) root-to-leaf path of edges whose outcomes all contain $q$.
  Call this path $q$'s \emph{search path}.
  Say that $q$ \emph{reaches} each node on this path.
  Call the leaf that $q$ reaches \emph{$q$'s leaf}.
  Define $q$'s \emph{depth} (in $T$) to be the depth of $q$'s leaf
  (equivalently, the number of tests on $q$'s search path).
  The \emph{cost} of $T$ is the weighted sum of the depths of all queries in $Q$
  (where each query $q\in Q$ has weight $w(q)$).
  A solution for $I$ is a decision tree for $I$ of minimum cost.

  A decision tree $T$ is called \emph{irreducible} if, for each node $u$ in $T$, (i) at least one query in $Q$ reaches $u$,
  and (ii) if any class $c\in \calC$ contains all the queries that reach $u$, then $u$ is a leaf.
  Any decision tree can easily be converted into an irreducible tree without increasing its cost,
  so we generally restrict attention to irreducible trees.
  As we shall see, in an irreducible tree $T$,
  each non-leaf node $u$ has a distinct test $X_u$
and
  each edge $u\to v$ has a distinct outcome,
  so, when convenient, we identify each node $u$ with its test $X_u$
  and identify each edge $u\to v$ with its outcome ($X_u$ or $\overline X_u$).

 Note that an {\LDT} instance is not necessarily \emph{feasible}, that is, it might not have a decision tree.
To be feasible, in addition to each query belonging to some class,
it must have the property that each set of queries that cannot be separated
by tests in $\calF$ (that is, for each test $X\in\calF$ either this set is a subset of $X$
or is disjoint with $X$) must be contained in some class.

An \emph{equality test with key $k$} is the test (set) $\{k\}$.
A \emph{less-than test with key $k$} is the test (set) $\{q\in Q: q < k\}$.
The \twoWCDT problem is the restriction of \LDT
to instances in which, for some set $K\subseteq Q$ of \emph{keys},
$\calF$ is comprised of the equality and less-than tests whose keys are in $K$.
(It is straightforward to verify that this is a laminar family.)
In this context we denote the instance as $I=(Q, w, C, K)$.

\medskip

The assumption $K\subseteq Q$ is for ease of presentation.
Also, we can assume without loss of generality that each query belongs to some class,
so $m \ge n = |Q|$ and the input size\footnote
{Note that $\calF$, being laminar, can be encoded as a tree in space $O(n)$.}
is $\Theta(n+m) = \Theta(m)$.
As discussed at the end of Section~\ref{sec: algorithm},
although our definition of \twoWCDTs allows only less-than and equality tests,
all results extend easily to the other standard inequality tests.

\paragraph{Successful-queries variants.}
Conventionally,
in the \emph{successful-queries} variants of binary search-tree problems, the input is an ordered set $K$ of weighted keys.
Each comparison must compare the given query value to a particular key in $K$
and each query must be a value in $K$.
Such queries are called \emph{successful}.
In the \emph{standard} variants, the input is augmented with a weight for each open interval between consecutive keys
(and before the minimum key and after the maximum key).
\emph{Unsuccessful} queries, that is, queries to values within these intervals, are also allowed.
They must be answered by returning the interval in which the query falls.
Our definition of \twoWCDTs captures both variants:
restricting to $Q=K$ gives the successful-queries variant,
while the standard variant can be modeled by adding one non-key query within each open interval to $Q$.



\section{Imbalance theorem for trees with laminar tests}\label{sec: imbalance}

This section states and proves Theorem~\ref{thm: imbalance} (the imbalance theorem):


\begin{restatable}{theorem}{thmImbalance}
  \label{thm: imbalance}%
  Let $T$ be any optimal, irreducible tree for an \LDT instance $I=(Q, w, \calC, \calF)$.
  Let $\uA\to \uB\to \cdots \to \uD$ be the downward path
  from any node $\uA$ to any proper descendant $\uD$ in $T$
  such that $\weight{\uB'} < \weight{\uD}$.
  Then the outcomes leaving $\uA\to\cdots\to\uD$
  are pairwise disjoint.
\end{restatable}

The outcomes leaving $\uA\to\cdots\to\uD$ are $\{\uito' : 1 \le i < d\}$.
Note that this does not include any outcome out of $\uD$.
Recall that in $T$ each node $u$ is identified with its test set $X_u$,
each edge $u\to v$ is identified with its outcome $X_u$ or $\overline X_u$,
and $u'$ denotes the sibling of $u$, unless $u$ is the root.

\paragraph{Intuition.}
The theorem considers how, in an optimal \twoWCDT,
it can happen that a node ($\uD$) can be heavier than the sibling ($\uB'$) of some ancestor ($\uB$).
If this happens, then it must be that we can't rotate the node up the tree above its ancestor.
The theorem says that this can happen only if the outcomes leaving the path from the ancestor to the node are disjoint.

Here are some examples to build intuition.
Let $T$ be as assumed in the theorem.
First suppose that each test along the path $\uA\to \cdots \to\uD$ with $\weight{\uB'} < \weight {\uD}$
is a less-than test.
Then each outcome leaving the path contains either $\min Q$ or $\max Q$,
so, by the theorem, at most two edges leave the path
(at most one containing $\min Q$, at most one containing $\max Q$).
That is, $d \le 3$.

Another consequence: for any downward path $x\to y\to z$ in $T$,
the weight $\weight {y'}$ of the sibling of $y$ is
at least $\min \big(\weight{z}, \weight{z'}\big)$.
(Otherwise,
applying the theorem to the path $x\to y \to z$,
and then to the path $x \to y \to z'$,
the outcome $x \to y'$ is disjoint from outcomes $y \to z'$ and $y \to z$,
so the outcome $x \to y'$ would be empty,
contradicting the definition of \LDTs.)

Finally, for any equality test $\test = k$ in $T$,
for any proper ancestor $a$ of $\test = k$,
the weight $\weight {a'}$ of the sibling of $a$ (if there is one) is at least $\weight k$.
(Otherwise, let $p$ be the parent of $a$.
Let $L_k$ be the yes-child of $\test = k$.
Then the theorem applies to the path $p\to a \to \cdots \to \test = k \to L_k$,
so the outcome of $p \to a'$
is disjoint from the outcome of $\test = k \to L_k'$,
so must be a subset of the outcome of $\test = k \to L_k$, i.e., the singleton $\{k\}$.
So the outcome $p \to a'$ is either empty, contradicting the definition of \twoWCDTs,
or also $\{k\}$, contradicting the irreducibility of $T$.)
As a special case every equality-test ancestor $\test = h$ of $\test = k$
satisfies $\weight h \ge \weight k$.

In fact,
Theorem~\ref{thm: imbalance} generalizes the key structural theorems of~\cite{Anderson2002} and~\cite{chrobak_simple_2021} for \twoWCSTs.
For instance, the heaviest-first property of \twoWCSTs follows easily from the above paragraph.
Indeed, fix an optimal, irreducible \twoWCST tree $T$.
Assume without loss of generality that,
if the parent of any leaf $L_k$ in $T$ is a test node $\test = {h}$,
then $\weight{h} \ge \weight k$.
(Otherwise just change the parent to $\test = k$, making $L_k$ the yes-child.)
To show that $T$ has the heaviest-first property,
consider any test node $\test{=}{h}$ whose no-subtree has a leaf $L_k$ for a key $k$.
We will show $\weight k \le \weight h$.
In the case that $L_k$ is a child of $\test{=}{h}$, then the previous assumption implies $\weight k \le \weight h$.
So assume that $L_k$ is not a child of $\test{=}{h}$.
If the parent of $L_k$ is not already $\test{=}{k}$,
consider replacing that parent by $\test{=}{k}$, making $L_k$ the yes-child.
This preserves optimality and correctness.
Now $\weight h \ge \weight k$ follows from the last sentence in the previous paragraph, applied to the (possibly modified) tree.


\paragraph{The generalized rotation.}
Next we lay the groundwork for the proof of Theorem~\ref{thm: imbalance}.
Fix an \LDT instance $I=(Q, w, \calC, \calF)$.
Say tests $X,Y\in \calF$ are \emph{equivalent} if $X = Y$ or $X = \overline Y$.
We'll use only the following property of $\calF$,
which is essentially\footnote
{Property~\ref{prop: laminarity X Y} is a-priori weaker than laminarity,
  but any family $\calF$ with Property~\ref{prop: laminarity X Y}
can be converted into an equivalent laminar family $\calF'$
  by fixing any element $q_0\in Q$
  and taking $\calF' = \{X\in \calF : q_0 \not\in X\} \cup \{\overline X : X\in \calF, q_0\in X\}$.}
a restatement of laminarity:

\begin{property}
  \label{prop: laminarity X Y}
  Given two non-equivalent tests $X, Y\in \calF$,
  among the four pairs of outcomes
  in $\{X, \overline X\} \times \{Y, \overline Y\}$,
  exactly one pair are disjoint.
\end{property}

Fix an irreducible tree $T$ for $I$.

\begin{property}
  \label{prop: laminarity}
  Let $u$ and $v$ be distinct non-leaf nodes in an irreducible decision tree $T$ for $I$.
  Then
  (i) the tests at $u$ and $v$ are not equivalent.
  If $u$ is a proper ancestor of $v$ then
  (ii)
  the outcome from $u$ on the path from $u$ to $v$
  overlaps with both outcomes from $v$,
  while
  (iii)
  the other outcome from $u$
  (the one leaving the path from $u$ to $v$)
  is a subset of one outcome from $v$,
  and disjoint from the other outcome from $v$.
\end{property}

\begin{proof}
  Part~(ii) follows directly from the irreducibility of $T$.
  In the case when $u$ is an ancestor of $v$,
  Part~(ii) implies both Part~(i) and, using Property~\ref{prop: laminarity X Y}, Part~(iii),
  so we are done in this case.
  Since Part~(i) (non-equivalence) holds when $u$ is an ancestor of $v$,
  it also holds when $v$ is an ancestor of $u$, just by reversing their roles.
  To finish we show Part~(i) when neither is an ancestor of the other.

  Suppose for contradiction that $u$ and $v$ are equivalent.
  Let $a$ be the lowest common ancestor of $u$ and $v$.
  Let $a\to b$ and $a\to b'$ be the outcomes from $a$ leading towards $u$ and $v$, respectively.
  Part (ii) holds for $a$ and $u$,
  so $a\to b$ overlaps both outcomes from $u$.
  By the same reasoning (reversing the roles of $u$ and $v$)
  outcome $a\to b'$ overlaps both outcomes from $v$,
  implying (by the equivalence of $v$ and $u$)
  that $a\to b'$ overlaps both outcomes from $u$.
  So both outcomes at $a$ overlap both outcomes at $u$,
  contradicting Property~\ref{prop: laminarity X Y}.
\end{proof}


Here is some hopefully mnemonic terminology:

\begin{definition}
  \label{def: preferred}
  Given an outcome $b\to c'$ in $T$ from a non-root node $b$ to child $c'$,
  let $a$ be the parent of $b$.
  Call the sibling $b'$ of $b$ the \emph{uncle} of the child $c'$.
  If $b\to c'$ is the outcome at $b$ that is disjoint from the outcome $a\to b'$ from the grandparent to the uncle,
  say that the child $c'$ and the outcome $b\to c'$ are \emph{preferred} by $b$.
\end{definition}

By Property~\ref{prop: laminarity}(iii),
$b$ has exactly one preferred child and one preferred outcome, which leads to that child.
Also, the outcome $a\to b'$ from the grandparent to the uncle
is a subset of the non-preferred outcome $b\to c$ at $b$.


\begin{definition}
  \label{def: rotation}
  Given a non-root test node $b$, let $a$ be the parent of $b$.
  \emph{Rotating $b$ (above $a$)}
  replaces the subtree $T_a$ rooted at $a$ in $T$
  with the subtree $T'_b$ obtained from $T_a$ as shown in Figure~\ref{fig: rotation},
  that is, it exchanges the nodes $a$ and $b$ along with
  the subtrees rooted at their respective children $b'$ and $c'$.
\end{definition}

\begin{figure}
  \centering
  \input{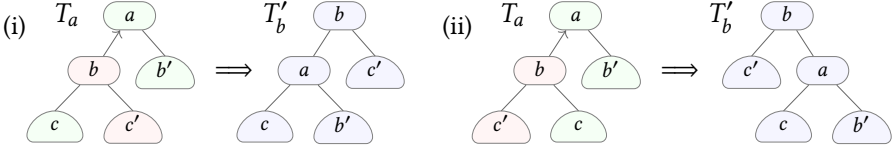}

  \caption{
    Rotating a non-root test node $b$ in $T$ moves $b$
    (along with its preferred child $c'$ and the subtree rooted at $c'$)
    above its parent $a$.
    Unlike binary search trees, laminar search trees are not inherently ordered.
    When drawing a rotation in a laminar tree, we draw the first tree $T_a$ using any convenient order,
    then, when drawing the rotated tree $T'_b$, order each node's two outcomes the same as they were ordered in $T_a$.
    Above, (i) and (ii) are two ways of drawing the exact same rotation.
    Throughout, $u'$ denotes the sibling of a given node $u$ in the original tree $T$,
    in which $T_a$ is a subtree.
  }\label{fig: rotation}
\end{figure}

%
Next we show that the rotation operation is correct.
To avoid confusion, note that, when considering a sequence of trees derived from $T$,
the notation $u'$ always denotes the sibling of node $u$ \emph{in $T$},
which is not necessarily the sibling of $u$ in subsequent trees.
Likewise, the notation $u\to v$ always denotes the outcome leading from $u$ to $v$ \emph{in $T$}.
The notation $u \stackrel{T'} \to v$ denotes the outcome leading from $u$ to $v$ in some subsequent tree $T'$.


\begin{observation}
  \label{obs: rotation}
  Let $T'$ be obtained from $T$ by rotating $b$ up as described above.
  Then (i) $T'$ is an irreducible decision tree for $I$,
  and (ii) the cost of $T'$ is the cost of $T$ plus $\weight{b'} - \weight {c'}$,
  so, provided $T$ is optimal, $\weight{b'} \ge \weight{c'}$.
  That is, the preferred child $c'$ cannot be heavier than its uncle $b'$.
\end{observation}

\begin{proof}
  \emph{Part (i).} Recall that the queries reaching a node
  are those in the intersection of all outcomes along the path from the root to the node.
  We will show that, for each leaf $L$, this set is the same in $T$ as it is in $T'$.

  If $L$ is not a descendant of $a$, the path from the root to $L$ does not change.
  If $L$ is a descendant of $c$, this path changes but the set of outcomes on this path is the same in $T$ and $T'$.
  It remains to consider the cases when $L$ is a descendant of $b'$ or $c'$ in $T$.

  In the case that $L$ is a descendant of $b'$,
  the only change to the path to $L$ is the addition of the outcome $b \to c$.
  (In $T'$ that outcome is now $b \stackrel {T'} \to a$.)
  But the path contains the outcome $a \to b'$,
  which (being disjoint from the preferred outcome $b\to c'$)
  is a subset of the non-preferred outcome $b \to c$.
  So the intersection is unchanged.

  Similarly, in the remaining case ($L$ is a descendant of $c'$)
  the path loses $a \to b$ (which is $a\stackrel {T'}\to c$ in $T'$).
  But the path contains the preferred outcome $b\to c'$
  which (being disjoint from $a\to b'$)
  is a subset of $a \to b$.
  So the intersection is unchanged.

  \smallskip
  \noindent
  \emph{Part (ii).}
  The rotation increases the depth of each descendant of the uncle $b'$ by one,
  while decreasing the depth of each descendant of the preferred child $c'$ by one,
  thus increasing the tree cost by $\weight{b'} - \weight {c'}$.
\end{proof}

Each non-root test node has a preferred child,
so by Observation~\ref{obs: rotation}
if $T$ is optimal each non-root test node has a child that weighs no more than the child's uncle:

\begin{observation}
  \label{obs: nephews}
  Suppose $T$ is optimal.
  For any non-root node $u$ with children $v$ and $v'$,
  $\weight {u'} \ge \min\big(\weight v, \weight {v'}\big)$.
\end{observation}


We now prove the theorem.

\begin{proof}[Proof of Theorem~\ref{thm: imbalance}]
  Let $T$, $I=(Q, w, \calC, \calF)$, and $\uA\to \uB\to \cdots \to \uD$ be as in the theorem statement,
  so $\weight{\uB'} < \weight{\uD}$.
  We claim that
  $\weight {\uB'} \ge \weight{\uC'} \ge \cdots \ge \weight {\uD'}.$
  Suppose otherwise for contradiction.
  Fix $j < d$ such that
  $\weight {\uB'} \ge \weight{\uC'} \ge \cdots \ge \weight {\uJ'} < \weight{u'_{j+1}}$.
  By Observation~\ref{obs: nephews} and $\weight {\uJ'} < \weight{u'_{j+1}}$,
  it must be that $\weight{\uJ'} \ge \weight{\uJp}$.
  Using this, the choice of $j$, and that $\uD$ is a descendant of $\uJp$, we have
  $\weight{\uB'} \ge \weight {\uJ'} \ge \weight {\uJp} \ge \weight{\uD}$,
  contradicting $\weight{\uB'} < \weight{\uD}$ and proving the claim.

  The claim, and $\weight{\uB'} < \weight{\uD}$, and the ancestry relations imply
  \begin{equation}
    \label{eq: chain}
    \weight {\uD'} \le \weight{\uDm'} \le \cdots \le \weight {\uB'}
    ~<~ \weight{\uD} \le \weight{\uDm} \le \cdots \le \weight{\uA}.
  \end{equation}

  Next suppose for contradiction that at least one pair of outcomes leaving the path overlaps.
  Fix such a pair $\up\to\upp'$ and $\uq\to\uqp'$ with $p < q < d$
  such that the later outcome $\uq\to\uqp'$ overlaps the earlier outcome $\up\to\upp'$,
  but \emph{is disjoint from} each outcome leaving the path between these two.
  (Formally, $\uq\to\uqp'$ overlaps $\up\to\upp'$
  but is disjoint from each $\ui\to\uip'$ with $p < i < q$.
  Such a pair must exist.
  For example, fix any $q < d$
  such that there is an earlier outcome leaving the path that overlaps $\uq\to\uqp'$.
  Then, among the latter, take $\up\to\upp'$ to be the one with maximum $p$.)

  \begin{figure}[t]
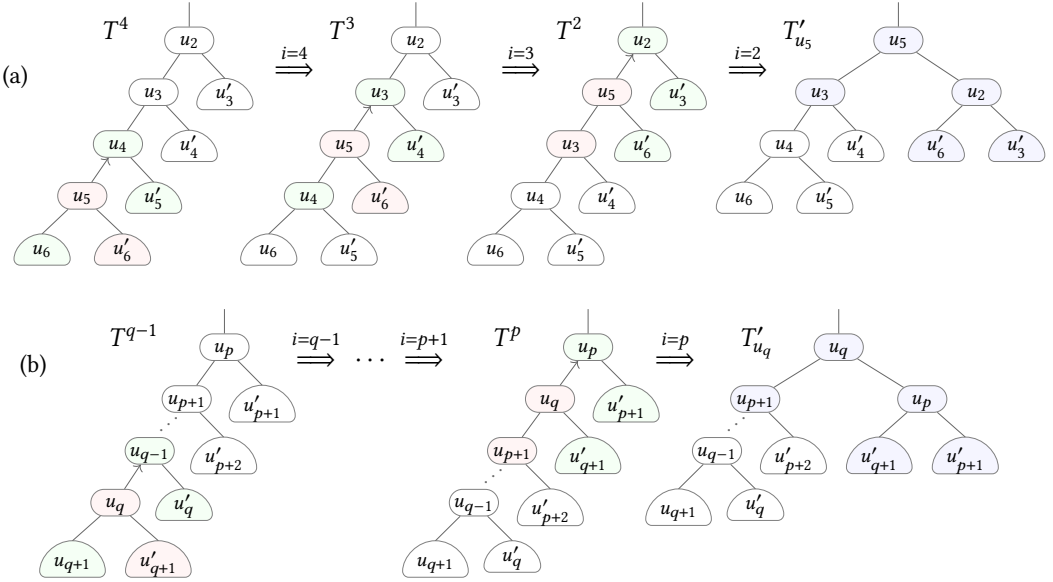

    \centering
  \input{FIGURES/TREES/imbalance_theorem}

    \vspace*{1em}

  \input{FIGURES/TREES/imbalance_theorem_general}

    \caption{
      The sequence of rotations in the proof of Theorem~\ref{thm: imbalance}.
      The drawing orders the initial tree $T^{q-1}=T$ so the path $\up\to\cdots\to\uq$ lies on the left spine.
      The case $(p, q) = (2, 5)$ is shown in (a).
      For the general case, (b) shows the first and last two trees in the sequence.
      In each rotation except the last, the preferred outcome of $\uq$ is $\uq\to\uqp'$.
      The preferred outcome is drawn to the right, so the rotation is of the form shown in Figure~\ref{fig: rotation}(i).
      It moves $\uq$ (and the preferred outcome $\uq\to\uqp'$) above $\ui$.
      Finally, in the last rotation, the preferred outcome of $\uq$ is $\uq\to\uqp$.
      The preferred outcome is drawn to the left,
      so the rotation is of the form shown in Figure~\ref{fig: rotation}(ii).
      This rotation moves the root $\up$ down and out of the path.
    }\label{fig: imbalance theorem}.
  \end{figure}

  Now, as illustrated in Figures~\ref{fig: imbalance theorem}(a) and~(b),
  rotate $\uq$ up the sub-path $\up\to\upp\to\cdots\to\uq$,
  ancestor by ancestor, just until $\uq$ becomes the parent of $\up$.
  That is, let $T^{q-1}$ be the initial tree $T$,
  then, for each $i \gets q-1, q-2, \ldots, p$ in decreasing order,
  let the next tree $T^{i}$ be obtained from the previous tree $T^{i+1}$
  by rotating $\uq$ above $\ui$.
  In each tree $T^i$ except the last, the parent of $\uq$ is $\ui$.
  The final tree $T^{p+1}$ is obtained from $T^p$ by rotating $\uq$ above $\up$.

  For each rotation except the last (each $i > p$),
  by the choice of $q$ and $p$,
  the outcome leaving $\uq$ that is disjoint from $\ui\to \uip'$ is $\uq\to\uqp'$
  (in both the original tree $T$ and the current tree $T^i$).
  So $\uq\to\uqp'$ is the preferred outcome for this rotation,
  and the rotation is as illustrated in Figures~\ref{fig: imbalance theorem}(a) and~(b).
  The preferred outcome is drawn to the right,
  so takes the form shown in Figure~\ref{fig: rotation}(i).
  It moves $\uq$ (and the preferred outcome $\uq\to\uqp'$) above $\ui$.
  Thus, just before the final rotation, the tree ($T^p$)
  is as shown in Figures~\ref{fig: imbalance theorem}(a) and (b),
  with $\uq$ (and the preferred outcome $\uq\to\uqp$) just below $\up$.
  (The tree $T^p$ could also be obtained directly from $T$
  by just deleting the three edges in $\uqm\to\uq\to\uqp$ and $\up\to\upp$
  and replacing them by the three edges $\uqm\to\uqp$ and $\up\to\uq\to\upm$.)
  The final rotation then rotates $\uq$ above $\up$.
  By the choice of $p$,
  the preferred outcome at $\uq$ for this rotation is $\uq\to\uqp$
  (in $T$; in the current tree $T^p$ this outcome is $\uq\stackrel {T^p}\to\upp$).
  So the rotation is as illustrated on the right of Figure~\ref{fig: imbalance theorem}(a) and (b),
  where the preferred outcome is drawn as the \emph{left} outcome of $\uq$,
  so is drawn in the form shown in Figure~\ref{fig: rotation}(ii).
  This rotation moves $\up$ down and out of the path.

  By inspection of the first and last trees in Figure~\ref{fig: imbalance theorem}(b),
  rotating $\uq$ (with $\uqp'$) all the way up the path and then rotating $\up$ out of the path in the final rotation
  changes the leaf depths as follows.
  The depths of descendants of $\uqp$ decrease by one, as they lose the ancestor $\up$.
  The depths of descendants of $\uqp'$ decrease by $q-p-1 \ge 0$, as they lose ancestors $\upp, \ldots, \uqm$.
  The depths of descendants of $\upp'$ increase by one, as they gain the ancestor $\uq$, which is rotated above them.
  The depths of other leaves in the subtree $T_{\uq}$ don't change, as they gain ancestor $\uq$ but lose $\up$.
  Hence, the increase in cost is at most $\weight{\upp'} - \weight{\uqp}$.
  From the optimality of $T$ it follows that $\weight{\upp'} \ge \weight{\uqp}$,
  contradicting~\eqref{eq: chain} and proving Theorem~\ref{thm: imbalance}.
\end{proof}


\section{Structural theorem for \twoWCDTsintitle}\label{sec: intermediate}

This section proves Theorem~\ref{thm: intermediate}, below,
which is an intermediate step towards proving the existence of an admissible tree.
The proof uses Theorem~\ref{thm: imbalance}, a ``bisection'' operation (a generalization of the rotation operation),
and specific properties of inequality and equality tests.
The example in Figure~\ref{fig: inversion} may be helpful in developing intuition for the theorem.

Let $T$ be an arbitrary irreducible tree for an arbitrary \twoWCDT instance $(Q, w, \calC, K)$.
Recall that, since we are working with classification rather than search,
the leaf $L_k$ for a key $k$ may have additional queries in its query set.


\begin{restatable}{theorem}{theoremIntermediate}
  \label{thm: intermediate}
  Suppose the instance has distinct weights and $T$ is optimal.
  Consider any equality-test node $\test = h$ in $T$ and a key $k$ with $\weight k > \weight h$ reaching this node.
  Then (i) a search for $h$ from the no-child of $\test = h$ would end at the leaf $L_k$ for $k$,
  and (ii) the path from $\test = h$ to $L_k$ has at most four nodes
  (including $\test = h$ and $L_k$).
  (iii) Also, $h$ is not in the class that $T$ assigns to $k$.
\end{restatable}

\begin{proof}
  Let $\uA\to\uB\to\cdots\to\uD$ be the path from $\test = h$ to $L_k$.
 As usual, the outcomes leaving the path are $\{\uito' : 1\le i<d\}$.
  So $\uA$ is $\test = h$, while $\uB'$ is the leaf for $h$, and $\uD$ is $L_k$.
  As $\weight{\uD} = \weight {L_k} \ge \weight k > \weight h = \weight{\uB'}$,
  the imbalance theorem (Theorem~\ref{thm: imbalance}) applies to the path.
  The theorem implies the following observation:


  \begin{observation}
    \label{obs: ast}
    The outcomes leaving the path are pairwise disjoint.
  \end{observation}

  The yes-outcome $\uAto'$ of $\test = h$ leaves the path, so by Observation~\ref{obs: ast}
  that outcome, that is, $\{h\}$,  is disjoint with all other outcomes leaving the path.
  Hence, a search for $h$ starting from the no-child $\uB$ of $\test= h$ would not leave the path,
  so would end at $L_k$.  This proves Part (i) of the theorem.

  \smallskip

  To prove Part (iii),
  suppose for contradiction that $h$ is in the class that $T$ assigns to $k$.
Then, in the case $d=2$, we could replace the node $\test{=}{h}$ by a leaf labeled with the class assigned by $T$ to $k$,
contradicting irreducibility. So assume $d\ge 3$.
  By Part (i) of the theorem, changing the test key at $\test = h$ to $k$ (and relabeling $\uB'$ with a class containing $k$)
  would give a correct tree, while decreasing the cost by $(\weight k - \weight h)(d-2)$.
  By assumption $\weight k > \weight h$,
  so $(\weight k - \weight h)(d-2) > 0$,
  and thus the modification would give a correct tree strictly cheaper than $T$,
  contradicting the optimality of $T$.

  \smallskip

  The rest of this section proves Part (ii), that is, that $d$ is at most $4$.
  Assume for contradiction that $d\ge 5$.
  We prove two independent lemmas.


  \begin{lemma}
    \label{lemma: 2 w(h) < w(uC)}~$2\,\weight h  < \weight{\uC}$.
  \end{lemma}

  \begin{proof}
    Consider inserting a new equality-test $\test = k$ above $\uC$,
    that is, replacing $T_{\uC}$ by a new equality test $\test =   k$
    whose yes-child is a new leaf labeled with any answer that $k$ accepts,
    and whose no-subtree is a copy of $T_{\uC}$.
    This increases the search depth of every query reaching $\uC$, except key $k$,  by 1.
    It decreases the search depth of $k$ by at least 1.
    Thus, the increase in cost is at most $(\weight{\uC} - \weight  k) - \weight  k$.
    With the optimality of $T$ this implies $\weight{\uC}  \ge 2\, \weight  k > 2\, \weight h $.
  \end{proof}

  \newcommand{\ambOp}{{\scriptstyle?}}

  Let $k_1 \le k_2 \le k_3 \le k_4$ be the comparison keys of the four tests in $\uA$, $\uB$, $\uC$, and $\uF$,
  sorted into non-decreasing order.
  Next we consider ``bisecting'' the subtree $T_{\uA}$ by introducing test node $\test < {k_3}$ as a new root
  and adjusting the rest of the tree appropriately.


  \begin{lemma}
    \label{lemma: bisectable}
    Among the four outcomes $\uito'$ (with $1\le i \le 4$) leaving the path,
    two are disjoint with the yes-outcome of $\test < {k_3}$,
    while the other two are disjoint with the no-outcome of  $\test < {k_3}$.
  \end{lemma}

  \begin{proof}
    We will show that $k_3$ has the desired property.

    Suppose at least two of the four tests in $\uA$, $\uB$, $\uC$, and $\uF$ are inequality tests,
    say $\test < {k_i}$ and $\test < {k_j}$ with $i < j$.
    Then (using $k_i \le k_j$) the yes-outcome of $\test < {k_i}$ and the no-outcome of $\test < {k_j}$ are disjoint.
    By Property~\ref{prop: laminarity} all other pairs of outcomes between the two nodes overlap.
    So, by Observation~\ref{obs: ast},
    \emph{if there are two less-than tests $\test < {k_i}$ and $\test < {k_j}$ in $\{\uA, \uB, \uC, \uF\}$ with $i<j$,
      then the outcomes leaving the path from  $\test < {k_i}$ and $\test < {k_j}$ are, respectively,
      the yes-outcome and the no-outcome.}

    By the preceding sentence, $\{\uA, \uB, \uC, \uF\}$ contains at most two less-than tests,
    and therefore at least two equality tests.
    The yes-outcome of any equality test $\ui$ is disjoint with some outcome of any $\uJ$,
    so by Property~\ref{prop: laminarity}
    the no-outcome of $\ui$ overlaps both outcomes of any $\uJ$ with $j\ne i$,
    and by Observation~\ref{obs: ast}
    \emph{the outcomes leaving the path from the (at least two) equality tests are yes-outcomes.}

    \smallskip

    Suppose for contradiction that the yes-outcome of some less-than test $\test < {k_j}$ with $k_j \ne k_1$ leaves the path.
    By the conclusion of the second-to-last paragraph above, and by $k_1 < k_j$,
    the test with key $k_1$ cannot be a less-than test, so must be $\test = {k_1 }$.
    But then the yes-outcome of $\test = {k_1}$ overlaps the yes-outcome of $\test < {k_j}$, contradicting Observation~\ref{obs: ast}.
    So \emph{if a yes-outcome leaves the path from any less-than test, the test's key is $k_1$.}
    By symmetric reasoning,
    \emph{if a no-outcome leaves the path from any less-than test, the test's key is $k_4$.}

    It follows that any inequality test in $\{\uA,\uB,\uC,\uF\}$
    must be in $\uA$ and/or $\uF$,
    implying that $\uB$ and $\uC$ do equality tests, so $k_2 < k_3$.

    For all $q\ge k_3$, none of the following hold: $q = k_1$, $q = k_2$ (using here $k_2 < k_3$), or $q < k_2$.
    So the no-outcome of $\test < {k_3}$ is disjoint with the yes-outcomes of $\test = {k_1}$,  $\test = {k_2}$,  and $\test < {k_1}$.
    By the conclusions of the preceding paragraphs,
    these include all outcomes that leave the path from the nodes with keys $k_1$ and $k_2$.
    Similarly, for all $q < k_3$, none of the following hold: $q = k_3$, $q= k_4$, or $q \ge k_4$,
    so the yes-outcome of $\test < {k_3}$ is disjoint with
    all outcomes that leave the path from the nodes with keys $k_3$ and $k_4$.
    This proves Lemma~\ref{lemma: bisectable}.
  \end{proof}

  Returning to the proof of Theorem~\ref{thm: intermediate}(ii), consider replacing $T_{\uA}$ in $T$
  by the subtree $T'$ obtained by \emph{bisecting} $T_{\uA}$ around the new node $\test < {k_3}$,
  in the following two steps (shown in Figure~\ref{fig: bisect0}).
  First, make a subtree with root $\test < {k_3}$, whose yes- and no-subtrees
  are each a copy of $T_{\uA}$.
  (Note that this subtree is a correct replacement for $T_{\uA}$.)
  For each outcome $\uito'$ ($1\le i\le 4$) that leaves the path $\uA\to\cdots\to\uFp$,
  per Lemma~\ref{lemma: bisectable}, the outcome is disjoint with
  either the yes-outcome or the no-outcome of  $\test < {k_3}$.
  If the outcome $\uito'$ is disjoint with the yes-outcome,
  \emph{splice it out} from the yes-copy of $T_{\uA}$.
  Otherwise (it is disjoint with the no-outcome)
  splice it out from the no-copy of $T_{\uA}$.

  Specifically, to \emph{splice out} the copy of  $\uito'$ means to
  remove that copy of $\ui$ and the subtree rooted at its child $\uip'$
  by replacing the subtree rooted at $\ui$
  by the subtree rooted at the current sibling of $\uip'$ (the other child of $\ui$),
  as happens in Figure~\ref{fig: bisect0}.
  The outcome from $\test{<}{k_3}$ that leads towards this copy of $\ui$
  is disjoint with the deleted outcome $\uito'$,
  so every search that reached the (now spliced out) copy of $\ui$
  continued through the sibling,
  so splicing out this copy of $\uito'$ preserves correctness.

  \begin{figure}[t]
    \centering
  \input{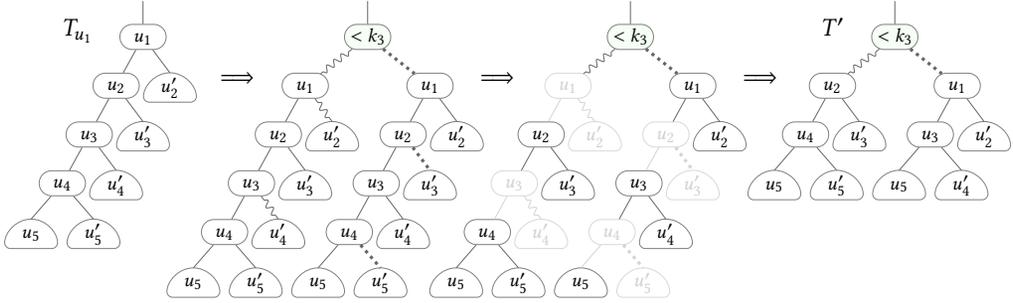}

    \caption{
      Bisecting $T_{\uA}$ around $\test < {k_3}$
      makes a new root $\test < {k_3}$,  makes each of its subtrees (yes and no) a copy of $T_{\uA}$,
      and then, for each outcome $\uito'$ leaving the path,
      splices out whichever copy of $\uito'$ is disjoint with the outcome of $\test < {k_3}$ that leads to that copy.
      (In the example here the squiggly outcomes are pairwise disjoint and the dashed outcomes are pairwise disjoint.)
      The outcomes on $\uA\to\cdots\to\uFp$ are drawn to the left. 
    }\label{fig: bisect0}
  \end{figure}

  By Lemma~\ref{lemma: bisectable}, two of the four outcomes are spliced out of the yes-copy of $T_{\uA}$,
  while the other two are spliced out of the no-copy, so
  the tree $T'$ obtained by bisecting $T_{\uA}$ around $\test < {k_3}$
  has one of the three forms shown in Figure~\ref{fig: bisect}(a), (b), or (c).
  Note that $T'$ contains two copies of the subtree $T_{\uFp}$ rooted at $\uFp$,
  so is not (in general) irreducible.
  However, it is still correct.

  Now we consider two cases, each reaching the desired contradiction.

  \begin{figure}[t]
    \centering
  \input{FIGURES/TREES/bisect}

    \caption{
      Bisecting $T_{\uA}$ around $\test < {k_3}$ yields a tree with one of forms (a), (b), or (c).
      The outcomes on the path $\uA\to\cdots\to\uD$ are drawn to the left,
      as is the outcome $\test < {k_3} \to \uA$.
    }\label{fig: bisect}
  \end{figure}

  \medskip
  \mycase{Case 1: The tree $T'$ has the form in Figure~\ref{fig: bisect}(a)}
  By inspection, the replacement increases the cost by
  $\weight{\uC'} + \weight{\uB'} - \weight{\uFp} - \weight{\uFp'} - \weight{\uF'}
  =\weight{\uC'} + \weight{\uB'} - \weight{\uC}$.
  By the optimality of $T$ this is non-negative.
  With Lemma~\ref{lemma: 2 w(h) < w(uC)} and $\weight h = \weight {\uB'}$ this implies
  $\weight{\uC'} \ge \weight{\uC}  -  \weight{\uB'} > \weight{\uB'}$.
  But then, by Theorem~\ref{thm: imbalance} applied to the path $\uA\to\uB\to\uC'$,
  the outcomes $\uAto'$ and $\uBto$ (leaving that path) are disjoint.
  By Observation~\ref{obs: ast}, outcomes $\uAto'$ and $\uBto'$  are also disjoint,
  contradicting Property~\ref{prop: laminarity} for $\uA$ and $\uB$.

  \medskip
  \mycase{Case 2: The tree $T'$ has one of the forms in Figure~\ref{fig: bisect}(b) or (c)}
  By inspection, either replacement increases the cost by
  $\weight{\uB'} - \weight{\uFp} - \weight{\uFp'} = \weight{\uB'} - \weight{\uF}$.
  With the optimality of $T$ this implies $\weight{\uB'} \ge \weight{\uF}$,
  which implies $\weight h \ge \weight k$, contradicting $\weight k > \weight h$.
  This proves Theorem~\ref{thm: intermediate}.
\end{proof}




\section{Some optimal tree is admissible}\label{sec: admissible}

{}\label{sec: admissible proof}


This section defines \emph{admissible} (Definition~\ref{def: admissible}),
then proves that some optimal tree is admissible (Theorem~\ref{thm: admissible}).
As mentioned in the introduction,
we first handle the case when all weights are distinct
(Lemma~\ref{lemma: admissible distinct})
then use a perturbation argument to extend to the general case.
The perturbation argument requires a globally consistent tie-breaking for equal-weight keys.

Let $T$ be any irreducible tree for a feasible \twoWCDT instance $I=(Q, w, \calC, K)$.

\begin{definition}[ordering queries by weight]
  For any query subset $R\subseteq Q$ and integer $i\ge 0$
  define $\heaviest i {R}$ to contain the $i$ heaviest queries in $R$  (or all of $R$ if $i\ge |R|$).
  For $q\in Q$, define $\heavier q$ to contain the queries (in $Q$) that are heavier than $q$.
  Define $\lighter q$ to contain the queries (in $Q$) that are lighter than $q$.
  Break ties among query weights arbitrarily but consistently throughout.
\end{definition}

Formally, we use the following notation to implement the tie-breaking mentioned above.
Fix an ordering of $Q$ by increasing weight, breaking ties in favor of
queries that are smaller in the linear ordering of $Q$.
(This particular tie-breaking rule is only for concreteness.  Any consistent rule would work.)
For $q\in Q$ let $\bweight q$ denote the rank of $q$ in this sorted order.
Throughout, given distinct queries $q$ and $q'$, define $q$ to be lighter than $q'$ if $\bweight q < \bweight {q'}$
and heavier otherwise ($\bweight q > \bweight {q'}$).
So, for example $\heaviest i {R}$ contains the last $i$ elements in the ordering of $R$ by increasing $\bweight q$.
The symbol $\undefsymbol$ represents the undefined quantity $\arg\max \emptyset$.
Define $\bweight \undefsymbol = \weight \undefsymbol = -\infty$,
\,$\heavier \undefsymbol = Q$, and $\lighter \undefsymbol = \emptyset$.


\begin{definition}[intervals and holes]
  For any $\ell, r\in Q$, let $\qint \ell r$ and $\kint \ell r$
  denote the \emph{query interval} $\{q\in Q: \ell \le q \le r\}$
  and \emph{key interval} $\{k\in K: \ell \le k \le r\} = K\cap \qint \ell r$.

  Given any non-empty query subset $R \subseteq Q$,
  call $\qint {\min R} {\max R}$ the \emph{query interval} of $R$.
  Define $k^*(R)$ to be the heaviest key in $R$, if there is one
  (that is, $k^*(R)=\arg\max\{\bweight k : k\in K\cap R\}$).
  Define also $\holes {R} = \qint {\min R} {\max R}\setminus R$ to be the set of \emph{holes} in $R$.
  We say that a hole $h\in \holes {R}$ is \emph{light} if $\bweight h < \bweight {k^*(R)}$, and otherwise \emph{heavy}.

  The set of queries reaching a node $u$ in a tree $T$ is called $u$'s \emph{query set}, and denoted $Q_u$.
  The query interval, and light and heavy holes, for $u$ are defined to be those for $u$'s query set $Q_u$.
  Write $w(u)$ as a shorthand for $w(Q_u)$,
  where $w(R) = \sum_{q\in R} w(q)$ denotes the total weight in the query set $R \subseteq Q$.
\end{definition}

If $R$ contains no keys then $k^*(R)$ is $\undefsymbol$ (undefined),
so $\bweight{k^*(R)}$ is $-\infty$ and $R$ has no light holes.

Each hole $h\in \holes{Q_u}$ at a node $u$ in a tree $T$
must result from a failed equality test $\test = h$
at an ancestor $v$ of $u$ in $T$, so $h\in K$.
The hole is light if any heavier key (and therefore $k^*(Q_u)$) reaches $u$.
For example, in the optimal tree in Figure~\ref{fig: inversion}(a) (in which $K=[4]$\})
the query set $Q_{\test = 1}$ of node $\test = 1$
has light holes $3$ and $4$.
These are lighter than the heaviest key $k^*(Q_{\test = 1}) = 2$ reaching $\test{=}{1}$,
but (not coincidentally, as we shall soon see) are the two heaviest in the node's key interval \emph{minus $2$'s class}.
The light holes in the query set of $\test{=}{1}$'s (right) no-child are $1$, $3$, and $4$,
which are the three heaviest in the node's key interval minus $2$'s class.
The query sets of the nodes in the trees in Figure~\ref{fig: inversion}(b) and~\ref{fig: inversion}(c) have no light holes,
but these trees are not optimal.


\begin{definition}[admissible]\label{def: admissible}
  A non-empty query subset $R\subseteq Q$ is \emph{admissible} if
  the set of light holes in $R$ is empty or has the form
  \[
    \heaviest b {\, \setkint {R} \cap \lighter {k^*(R)}\setminus c \, }
  \]
  for some  $b\in [3]$  and $c\in \calC$ such that $k^*(R)\in c$.
  (Throughout, for $i\in\mathbb N$, let $[i]$ denote $\{1,2,\ldots, i\}$.)

  The tree $T$ (or any subtree) is \emph{admissible} if all its nodes have admissible query sets.
\end{definition}


By definition, the holes of any query set $R$
lie in $R$'s key interval $\setkint R$, and its light holes are those lighter than $k^*(R)$, the heaviest key in $R$.

We next show Lemma~\ref{lemma: admissible distinct}.
Here is the intuition.
We need to constrain how the heaviest-first property can fail at a node $u$ in $T$.
One way the property can fail (as illustrated in Figure~\ref{fig: inversion}(a)),
is that there is a single class $c$ that contains all of $Q_u$ except for a few scattered keys,
so that the optimal tree can use equality tests to pull out these ``stragglers'',
then use a single leaf (labelled with $c$) to handle the rest.
These stragglers can include a few keys lighter than $k^*(u)$,
whose removal creates light holes, violating the heaviest-first property.

In fact, the proof shows that the path from $u$ to such a leaf can have length at most four.
(The path may have up to two less-than tests.)
The lemma states that if $Q_u$ fails to be heaviest first (that is, $Q_u$ has light holes), it will still be admissible:
for some $b\in[3]$ and some class $c$ that can be assigned to $k^*(Q_u)$,
the light holes must be the $b$ heaviest keys in $R$'s interval that are lighter than $k^*(Q_u)$ and are not in $c$.
We can think of this as the heaviest-first property being preserved with respect to the keys \emph{minus those in $c$},
with the restriction that
at most three keys from $c$ can be
    exempted from being holes
in this way.
(This restriction to $O(1)$ keys is helpful for efficiency.)
As we see later, the number of possible admissible query sets will turn out to be small enough
to yield an efficient dynamic program.

As an exercise, consider the instance with query set $Q = [8]$, with classes and weights
as specified in the table below,
and key set $K = \braced{2,3,4,5,7}$.
(Keys are in underlined.)

\newcolumntype{C}[1]{>{\centering\let\newline\\\arraybackslash\hspace{0pt}}m{#1}}%
\newlength{\mycol}%
\setlength{\mycol}{0.35in}%
\begin{center}
  \begin{tabular}{|r|C{\mycol}|C{\mycol}|C{\mycol}|C{\mycol}|C{\mycol}|C{\mycol}|C{\mycol}|C{\mycol}|} \hline
    query & 1
    & \underline{2}
    & \underline{3}
    & \underline{4}
    &\underline{5}
    & 6
    & \underline{7}
    & 8
    \\ \hline
    classes	&$ C$
    &$ A, B$
    &$ B, D$
    &$ A, C$
    &$ B$
    &$ C$
    &$ C, D$
    &$ B$
    \\ \hline
    weight	& 10
    & 13
    & 67
    & 49
    & 27
    & 58
    & 38
    & 12
    \\ \hline
  \end{tabular}
  \smallskip
\end{center}

For the subset $R_1 = \braced{1,2,4,8}$, we have $k^*(R_1) = 4 \in A\cap C$.
Subset $R_1$ has three holes: a heavy hole $3$ and two light holes $5,7$.
In the above definition, choose $A$ for the class $c$ of $k^*(R_1)$.
Then $5$ and $7$ are the two heaviest keys in $ \setkint {R_1} \cap \lighter {k^*(R_1)} \setminus A$. So $R_1$ is admissible.
For the subset $R_2 = \braced{2,3,6,8}$, we have $k^*(R_2) = 3$, and three holes $4$, $5$ and $7$, all light.
Both classes ($B$ and $D$)
that contain $k^*(R_2)$ also contain one of the light holes, so $R_2$ is not admissible.

Perhaps counterintuitively, the admissibility of a set $R$ is \emph{not} determined solely by the subinstance
naively defined by $R$. (This instance is $I_R = (R, w_R, \calC_R, K_R)$,
where $w_R$ is $w$ restricted to $R$,
while $\calC_R$ is $\{c \cap R : c\in \calC\} \setminus \{\emptyset\}$, and $K_R$ is $K\cap R$.)
Admissibility of $R$ also depends on its set of light holes, in $K\setminus R$.
This will be important for the implementation.


\begin{lemma}\label{lemma: admissible distinct}
  If the instance has distinct weights and $T$ is optimal, then $T$ is admissible.
\end{lemma}

\begin{proof}
  \newcommand{\T}{T}
  Consider any node $u$ in $\T$.
  To prove the lemma we show that $u$'s query set is admissible.
  If $Q_u$ has no light holes, then we are done, so assume otherwise.
  Let $k^*=k^*(Q_u)$ be the heaviest key reaching $u$.
  Let $\lightholes_u=\holes {Q_u} \cap \lighter {k^*}$ be the set of light holes at $u$.
  Let $b=|\lightholes_u|$.
  Let $c$ be the class that $\T$ assigns to $k^*$ and $S = \setkint {Q_u} \cap \lighter {k^*}\setminus c$.
  We want to show $\lightholes_u = \heaviest b S$ and $b\in[3]$.

  First we show $\lightholes_u \subseteq S$.
  By definition, $\lightholes_u \subseteq \setkint {Q_u} \cap \lighter {k^*}$.
  For any light hole $h\in \lightholes_u$,  key $k^*$ is heavier than $h$ and reaches the ancestor $\test = h$ of $u$.
  Applying Theorem~\ref{thm: intermediate} to that ancestor, hole $h$ is not in $c$.
  It follows that $\lightholes_u \subseteq S$.

  Next (recalling $b=|\lightholes_u|$) we show \(\lightholes_u = \heaviest b S\).
  Suppose otherwise for contradiction.
  That is, there are $k\in S\setminus \lightholes_u\subseteq Q_u$ and $h\in \lightholes_u$ such that $k$ is heavier than $h$.
  Keys $k^*$ and $k$ reach the ancestor $\test = h$ of $u$.
  Applying Theorem~\ref{thm: intermediate} (twice) to that ancestor,
  the search path for $h$ starting from the no-child of $\test = h$
  ends both at $L_{k^*}$ and at the leaf $L_k$ for $k$.
  So $L_k = L_{k^*}$, which implies that $k$ is in $c$, contradicting $k\in S$.
  Therefore \(\lightholes_u = \heaviest b S\).

  Finally, we show that $b\le 3$.
  Let $h\in \lightholes_u$ be the light hole whose test node $\test = h$ is closest to the root.
  Key $k^*$ reaches $\test = h$ and weighs more than $h$.
  Applying Theorem~\ref{thm: intermediate} to $\test = h$ and key $k^*$,
  the path from $\test = h$ to $L_{k^*}$ has at most four nodes (including the leaf).
  Each light hole has a unique equality-test node on that path.
  So (using that $u$ is on this path)
  there are at most three light holes in $Q_u$.
\end{proof}


Now we use a perturbation argument to extend Lemma~\ref{lemma: admissible distinct} to the general case.
Recall that ``feasible'' means the instance has a correct tree.
As discussed in Section~\ref{sec: introduction}, not all instances do.

\begin{restatable}{theorem}{thmWellBehaved}\label{thm: admissible}
  If the instance is feasible, then some optimal tree is admissible.
\end{restatable}
\begin{proof}
  Assume the instance $I=(Q, w, \calC, K)$ is feasible.
  Recall that $\bweight q$ is the rank of $q$
  in the sorting of $Q$ by weight, breaking ties consistently,
  as defined at the start of the section.

  Let $I^*=(Q, \pweightName, \calC, K)$ be an instance obtained from $I$
  by perturbing the query weights infinitesimally so that (i) the perturbed weights are distinct
  and (ii) sorting $Q$ by $\pweightName$   gives the same order as sorting by $\bweightName$.
  Specifically, take $\pweight q = \weight q + \delta\, \bweight q$,
  for $\delta$ such that $0<\delta < \epsilon/n^3$,
  where $\epsilon>0$ is the minimum of two
  quantities:
  the minimum absolute difference between any two distinct weights
  and
  the minimum absolute difference in cost
  between any two irreducible trees with distinct costs,
  using here that there are finitely many irreducible trees.
  Recall also that $\bweight q \in [n]$.

The concept of tree irreducibility (defined in Section~\ref{sec: definitions}) is independent
of the weight function ($\weightName$ or $\bweightName$).
So the sets of irreducible trees for $I$ and for $I^*$ are the same.

  Let $T^*$ be an optimal, irreducible tree for $I^*$ (so also irreducible for $I$).
  Applying Lemma~\ref{lemma: admissible distinct} to $T^*$ and $I^*$,
  tree $T^*$ is admissible for $I^*$.
  By inspection of Definition~\ref{def: admissible},
  whether $T^*$ is admissible for an instance
  depends only on $T^*$ and the (tie-broken) ordering of the queries by weight.
  Since these orderings are the same in $I$ and $I^*$,
  the tree $T^*$ is admissible for $I$ if and only if it is admissible for $I^*$.

  To finish we observe that $T^*$ is also optimal for $I$.
  For any tree $T'$, let $\cost{T'}$ and $\pcost{T'}$ denote the costs of $T'$
  under weight functions $w$ (for $I$) and $\pweightName$ (for $I^*$), respectively.
  Recall that earlier we fixed $T$ to be an irreducible tree for $I$.
  Assume that $T$ is also optimal for $I$.
  Then
  \begin{equation*}
    \cost {T^*}
    \;\le\; \pcost {T^*}
    \;\le\; \pcost {T}
    \le \cost {T} + n^3\delta
    \;<\; \cost {T} + \epsilon.
  \end{equation*}
  So by the choice of $\epsilon$ we have $\cost{T^*} \le \cost T$.
  Therefore $T^*$ is optimal for $I$ as well.
\end{proof}



\section{Algorithm}\label{sec: algorithm}

This section proves the main result:

\begin{restatable}{theorem}{thmAlgorithm}\label{thm: algorithm}
  There is an $O(n^3 m)$-time algorithm for finding a minimum-cost \twoWCDT.
\end{restatable}

\begin{proof}
  Fix the input, an arbitrary $\twoWCDT$ instance $I =  (Q, w, \calC, K)$.
  Let $\AdmSubsets$ denote the set of admissible query subsets of $Q$ (per Definition~\ref{def: admissible}).
  For any $R\in \AdmSubsets$,
  if $R$ is contained in some class,
  then the tree for $R$ consists of a single leaf (labeled with some such class).
  Otherwise an admissible tree for $R$
  consists of any root $u$ whose test partitions $R$ into $(\setyes{R}{u}, \setno{R}{u})$
  (the bipartition of $R$ into those values that satisfy $u$ and those that don't),
  with $u$'s yes-subtree being any admissible tree for $\setyes{R}{u}$
  and $u$'s no-subtree being any admissible tree for $\setno{R}{u}$.
  So, defining $\acost R$ to be the minimum cost of any subtree for $R$ that is admissible for $I$,\footnote{An observant reader
  may notice that it can be that $\acost R > \cost R$ (the minimum cost of any tree for $T$),
  but if so $R$ cannot actually occur as the query set of any node in an optimal tree.}
  the following recurrence holds:


  \begin{recurrence}\label{recurrence}
    For any $R\in \AdmSubsets$,
    \[
      \acost {R} =
      \begin{cases}
        0 & ((\exists c\in \calC)\, R\subseteq c) \\
        \weight{R} + \min_{u}\big(  \acost {\setyes {R} u} + \acost {\setno {R} u} \big), & (\textit{otherwise})
      \end{cases}
    \]
    where $u$ ranges over the allowed tests (defined in Section~\ref{sec: definitions})
    for which $\setyes{R}{u}$ and $\setno{R}{u}$ are in $\AdmSubsets$ (that is, admissible).
    If there are no such tests the minimum is infinite.
  \end{recurrence}


  The algorithm returns $\acost Q$, the minimum cost of any admissible tree for $I=(Q, w, \calC, K)$.
  By Theorem~\ref{thm: admissible}, this equals the minimum cost of any tree for $I$, so the algorithm is correct.
  Next we describe how to achieve the desired running time.


  There are $O(n^2 m)$ admissible query sets.
  (Indeed, for any admissible set $R$,
  if $R$ has no light holes it is determined by the triple $(\min R, \max R, k^*(R))$.
  Otherwise, per Definition~\ref{def: admissible},
  $R$ is determined by a tuple $(\min R, \max R, k^*(R), b, c)$,
  where $(b, c) \in [3] \times \calC$ with $k^*(R) \in c$.)
  So $O(n^2 m)$ subproblems arise in recursively evaluating $\acost Q$.
  To achieve the desired time bound,
  it suffices to evaluate the right-hand side of Recurrence~\ref{recurrence} for any given $R\in\AdmSubsets$ in $O(n)$ amortized time.
  Next we describe how to do this.

  Assume (by renaming elements in $Q$ in a preprocessing step) that $Q = [n]$.
  Given a non-empty query set $R\subseteq Q$, define the \emph{signature} of $R$ to be
  \[
    \tau(R) = (\min R, \max R, k^*(R), \lightholes(R)),
  \]
  where $\lightholes(R) = \holes {R} \cap \lighter {k^*(R)}$ is the set of light holes in $R$.

  For any $R$, its signature is easily computable in $O(n)$ time
  (for example, bucket-sort $R$ and then enumerate the hole set $\qint \ell r \setminus R$ to find $\lightholes(R)$).
  Each signature is in the set
  \[
    \calS = Q\times Q \times (K \cup \{\undefsymbol\}) \times 2^Q
  \]
  of \emph{potential signatures}.
  Conversely, given any potential signature $t=(\ell, r, k, H')\in\calS$,
  the set $\tau^{-1}(t)$ with signature $t$, if any, is unique and computable from $t$ in $O(n)$ time.
  Specifically, $\tau^{-1}(t)$ is equal to the query set $Q_{(t)} = \qint \ell r \setminus ( (K\cap \heavier   k)  \cup H')$,
  provided that $Q_{(t)}$ is non-empty and has signature $\tau(Q_{(t)}) = t$; otherwise
  $\tau^{-1}(t)$ is undefined.
  (In general, the signature of $Q_{(t)}$ may be different from $t$;
  for example we may have $k \notin \qint{\ell}{r}$, or  one of $\ell$, $r$ may be in $H'$.)


  To finish the proof we prove Lemma~\ref{lemma: time}:

  \begin{restatable}{lemma}{lemmaTime}\label{lemma: time}
    After an $O(n^3 m)$-time preprocessing step,
    given the signature $\tau(R)$ of $R\in\AdmSubsets$,
    the right-hand of Recurrence~\ref{recurrence} is computable in amortized time $O(n)$.
  \end{restatable}

  \begin{proof}
    Note that the admissible sets can be enumerated in $O(n^3 m)$ time as follows.
    First do the $O(n^3)$ admissible sets without light holes:
    for each $(\ell, r, k) \in  Q\times Q \times (K \cup \{\undefsymbol\})$,
    output $\tau^{-1}(\ell, r, k, \emptyset)$ if it exists.
    Next do the $O(n^2 m)$ admissible sets with at least one light hole,
    following Definition~\ref{def: admissible}:
    for each $(\ell, r, k, b, c) \in  Q\times Q \times K \times [3]\times \calC$ with $k\in c$,
    letting $H' =  \heaviest b  {\kint {\ell} {r} \cap \lighter   k\setminus c}$,
    if $H'$ is well-defined then output $\tau^{-1}(\ell, r, k, H')$ if it exists.

    The preprocessing step initializes the dictionary for admissible query subsets and identifies the leaves.
    Here are the details.


    \newcommand{\proofPara}[1]{\medskip\par\noindent\emph{#1}}

    \proofPara{Initialize a dictionary $\Dictionary$ holding a record $\Dictionary[\tau(R)]$ for each set $R$ in $\AdmSubsets$.}
    To be able to determine whether a given query set $R$ is in $\AdmSubsets$,
    and to store information (including the memoized cost) for each admissible set $R$,
    build a dictionary $\Dictionary$ that holds a record $\Dictionary[\tau(R)]$
    for each $R\in\AdmSubsets$, indexed by the signature $\tau(R)$.
    For now, assume the dictionary $\Dictionary$ supports constant-time access to the record $\Dictionary[\tau(R)]$
    for each $R\in\AdmSubsets$ given the signature $\tau(R)$ of $R$.
    (We describe a suitable implementation later.)
    Initialize $\Dictionary$ to hold an empty record $\Dictionary[\tau(R)]$ for each $R\in \AdmSubsets$
    by enumerating all $R\in \AdmSubsets$ as described above.
    This takes $O(n^3 m)$ time.


    \proofPara{Identify the leaves.}
    To identify the sets $R\in\AdmSubsets$ that are leaves    (that is, such that $(\exists c\in \calC)~R\subseteq c$)
    in $O(n^3 m)$ time,
    for each triple $(\ell, r, k) \in Q\times Q \times (K \cup \{\undefsymbol\})$, do the following two steps.

    \begin{enumerate}
		
    \item Let $\calR \subseteq \AdmSubsets$ contain the admissible sets $R$ such that $\tau(R) = (\ell, r, k, H')$ for some $H'$.
      Assume $\calR$ is non-empty (otherwise move on to the next triple).
      Let $R_\emptyset$ be the set with signature $(\ell, r, k, \emptyset)$,
      so that each $R\in\calR$ is a subset of $R_\emptyset$ and can be written as $R_\emptyset \setminus \lightholes(R)$.
      Let $\calC_\ell$ contain the classes $c\in\calC$ such that $\ell\in c$.
      Observe that $|\calR|  \le 4 |\calC_\ell|$, because $R_\emptyset$ is unique for the triple $(\ell, r, k)$,
      and then each $R\in\calR$ is determined from $R_\emptyset$  by the class $c\in \calC$ and the number $b\in [3]$ of light holes, per Definition~\ref{def: admissible}.

    \item Each set $R\in \calR$ contains $\ell$,
      so $R$ is a leaf if and only if $R\subseteq c$ for some $c\in \calC_\ell$.
      The condition $R\subseteq c$ is equivalent to $R_\emptyset \setminus \lightholes(R) \subseteq c$,
      which is equivalent to $R_\emptyset \setminus c \subseteq \lightholes(R)$.
      So, any given set $R\in \calR$ is a leaf if and only if some subset of $\lightholes(R)$
      equals $R_\emptyset \setminus c$ for some $c\in\calC_\ell$.
      Identify all such $R$ in time $O(n |\calR| + n |\calC_\ell|)$.
      (Recalling that $|\lightholes(R)|\le 3$ for each $R\in \calR$, this is straightforward.
      One way is to construct the collection $\calH = \bigcup_{R\in \calR} 2^{\lightholes(R)}$ of subsets of the light-hole sets.
      Order the elements within each subset in $\calH$ by increasing value,
      then radix sort $\calH$ into lexicographic order.
      Do the same for the collection $\calL = \{R_\emptyset \setminus c: c\in\calC_\ell, |R_\emptyset\setminus c|\le 3\}$.
      Then merge the two collections to find the elements common to both.
      A given $R\in\calR$ is a leaf if and only if some subset of $\lightholes(R)$ in $\calH$ also occurs in $\calL$.)

    \end{enumerate}

    As noted above, we have $|\calR| \le 4 |\calC_\ell|$,
    so the time spent above on a given triple $(\ell, r, k)$ is $O(n |\calR| + n |\calC_\ell|) = O(n|\calC_\ell|)$.
    Summing over all triples $(\ell, r, k)$,
    the total time is $O(n^2 \sum_{\ell\in Q} n |\calC_\ell|) = O(n^3 m)$.

    In $O(n^3 m)$ time, identify the $O(n^2 m)$ leaves $R\in\AdmSubsets$ as described above.
    For each, record in its entry $\Dictionary[\tau(R)]$ that $R$ is a leaf and that $\acost {R} = 0$.


    \proofPara{How to implement Recurrence~\ref{recurrence}.}
    Next we describe how to compute $\acost {R}$, given the signature $\tau(R)=(\ell, r, k, H')$
    of any set $R\in \AdmSubsets$, in $O(n)$ time.

    If the record $\Dictionary[\tau(R)]$ already holds a memoized cost for $R$, then we are done, so assume otherwise.
    (This implies that $R$ is not a leaf.)
    In $O(n)$ time, build $R$ from $\tau(R)$ and calculate the sum $\weight {R}$.
    Let $R=(q_1, q_2, \ldots, q_z)$ be $R$ in increasing order, computed using bucket sort.
    For every possible test node $u$,
    precompute the signatures $\tau(\setyes R u)$ and $\tau(\setno R u)$.
    Do this in two $O(n)$-time stages:
    one stage for all possible less-than tests,
    the other stage for all possible equality tests:

    \begin{description}

      \item{\textbf{Stage~1:}}
      Precompute the pair of signatures $\tau(\setyes R {\test < h})$
      and $\tau(\setno R {\test < h})$
      for every $h\in K$
      as follows:

      \begin{description}

        \item{1.1.}
        For $i \in \{0, 1, \ldots, z\}$,
        define $R_i = (q_1, q_2, \ldots, q_i)$ and $\overline R_i = (q_{i+1}, q_{i+2}, \ldots, q_z)$.
        For $h\in Q=[n]$,
        define $i(h)$ to be the index such that $\setyes R {\test < h}$ is $R_{i(h)}$
        and $\setno R {\test < h}$ is $\overline R_{i(h)}$.
        Precompute $i(h)$ for all $h\in Q$ in $O(n)$ total time.
        (Note that $i(h) = \max~ \{0\}\!\cup\! \{i \in [z] : q_i < h\}$.
        Take $i(1) = 0$, then for $i\gets 2, 3, \ldots, n$ take $i(h) = i(h-1) + 1$ if $q_{i(h-1)+1} < h$;
        otherwise take $i(h) = i(h-1)$.)

        \item{1.2.}
        Compute $k^*(R_i)$ and $k^*(\overline {R_i})$ for all $i$.
        (These are the heaviest keys in $R_i$ and in $\overline {R_i}$, respectively.
        First take $k^*(R_0) = \undefsymbol$,
        then, for $i\gets 1, \ldots, z$,
        take $k^*(R_i) = q_i$ if  $q_i\in K$ and $q_i$ is heavier than $k^*(R_{i-1})$,
        and otherwise $k^*(R_i) = k^*(R_{i-1})$.
        Take $k^*(\overline {R_i}) = k^*(R)$ if $i<z$ and $k^*(R) \ge q_{i+1}$; otherwise take $k^*(\overline {R_i}) = \undefsymbol$.)

        \item{1.3.} Compute the light-hole sets
        $\lightholes(R_i) = \{h\in H' : h \le q_i\}$
        and $\lightholes(\overline{R_i}) = \{h\in H' : h \ge q_{i+1}\}$.
        (Each such set can be computed in constant time from $H'$, as $|H'|\le 3$.)

        \item{1.4.} Finally, enumerate all $h\in K$.
        For each, compute the pair of signatures
        $\tau(\setyes R {\test < h})$
        and $\tau(\setno R {\test < h})$,
        using
        $\setyes R {\test < h} = R_{i(h)}$,
        $\setno R {\test < h} = \overline R_{i(h)}$,
        and (for $i=i(h)$),
        $\tau(R_i) = (q_1, q_i, k^*(R_i), \lightholes(R_i))$
        and
        $\tau(\overline {R_i}) = (q_{i+1}, q_z, k^*(\overline R_i), \lightholes(\overline R_i))$.
        (Given the results of the previous three steps, this takes constant time per $h$.)
      \end{description}

    \item{\textbf{Stage~2:}}
      Precompute the pair of signatures $\tau(\setyes R {\test = h})$
      and $\tau(\setno R {\test = h})$
      for every $h\in K\cap R$.
      For each such $h$, we have $\setyes {R} {\test = h} = \{h\}$, so $\tau(\setyes {R} {\test = h}) = (h, h, h, \emptyset)$,
      and $\tau(\setno R {\test = h})$ can be computed as follows:
      \begin{description}
      \item{2.1.}
        If $h\not\in \{\min R, \max R, k^*(R)\}$
        (using that $|R| \ge 2$, as $R$ is not a leaf, so $\setno {R} {\test = h} \ne \emptyset$)
        the signature $\tau(\setno {R} {\test = h})$ is $(\min R, \max R, k^*(R), H' \cup \{h\})$,
        which (as $|H'|\le 3$) is computable from $\tau(R)$ in constant time.
      \item{2.2.}
        Otherwise ($h$ is one of the three values in $\{\min R, \max R, k^*(R)\}$),
        using $\setno {R} {\test = h} = R \setminus \{h\}$,
        explicitly compute $\setno {R} {\test = h}$ and its signature in $O(n)$ time.
      \end{description}
    \end{description}

    Finally, for each pair of signatures $\tau(\setyes {R} u)$ and $\tau(\setno {R} u)$ enumerated above (in Step 1.4 or Stage 2),
    check whether $\setyes {R} u$ and $\setno {R} u$ are admissible
    (by checking, in constant time, whether their signatures have entries in $\Dictionary$).
    If so, compute the values of $\acost {\setyes {R} u}$ and $\acost {\setno {R} u}$
    recursively from their signatures.
    Then, for $\acost {R}$, returns (and memoize in $\Dictionary[\tau(R)]$)
    the value from the recurrence, namely
    $\weight{R} + \min_{u} (  \acost {\setyes {R} u} + \acost {\setno {R} u} )$,
    with the minimum taken over all such $u$.

    In this way, for each $R\in\AdmSubsets$, the time to evaluate the right-hand side of the recurrence is $O(n)$.
    There are $O(n^2 m)$ sets in $\AdmSubsets$, so the total time is $O(n^3 m)$.


    \proofPara{How to implement the dictionary $\Dictionary$.}
    For each admissible query set $R\in\AdmSubsets$, the set $\lightholes(R)$ of light holes has size at most three.
    It follows that the signature $\tau(R) = (\ell, r, k, \lightholes(R))$ has size $O(1)$
    and one way to implement the dictionary $\Dictionary$ (to support constant-time lookup)
    is to use a hash table with universal hashing.
    Then the algorithm uses space $O(n^2 m)$, but is randomized.
    If a deterministic implementation is needed,
    one can implement the dictionary by storing an $n\times n\times n$ matrix $\Matrix$ of buckets
    It follows that the signature $\tau(R) = (\ell, r, k, \lightholes(R))$ has size $O(1)$
    such that a given bucket $\Matrix[\ell, r, k]$
    holds the records for the admissible query sets $R$
    with signatures of the form $\tau(R) = (\ell, r, k, H')$ for some $H'$.
    Organize the records in this bucket using a trie (prefix tree) of depth 3
    keyed by the (sorted) keys in $H'$.
    This still supports constant-time access, but increases the space to $O(n^3 m)$.
    More generally, for any $d\ge 1$, one can represent each element $k\in [n]$ within each set $H'$
    as a sequence of $\lceil \log_2(n)/d\rceil$ $d$-bit words,
    then use a trie with alphabet $\{0,1,\ldots,2^d-1\}$ and depth at most $3\lceil \log_2(n)/d\rceil$.
    Then space is $\Theta(2^d n^2 m)$ while the access time is $\Theta(\log(n)/d)$.
    For example, we can take $d = \lceil \epsilon\log_2 n \rceil$ for any constant $\epsilon$
    to achieve space $O(n^{2+\epsilon} m)$
    and access time $\Theta(1/\epsilon) = \Theta(1)$.
    Or we can take $d=1$ and achieve space $O(n^2 m)$ and access time $\Theta(\log n)$,
    increasing the total time to $O(n^3 m \log n)$.
  \end{proof}

	Per Lemma~\ref{lemma: time}, the preprocessing takes time $O(n^3m)$,
	and for each of $O(n^2m)$ sets $R\in\AdmSubsets$ Recurrence~\ref{recurrence}
	can be evaluated in time $O(n)$.  This proves Theorem~\ref{thm: algorithm}.
\end{proof}

\paragraph{Remarks.}

In the common case that $\calC$ partitions $Q$, each query $q\in Q$ is contained in just one class $c\in \calC$
(so $m = n$ and the algorithm runs in time $O(n^4)$), and then
the algorithm can be implemented to use space $O(n^2 m) = O(n^3)$.
To do this, in the above implementation of the dictionary using a matrix $\Matrix$ of buckets,
each bucket $\Matrix[\ell, r, k]$ stores the records of at most four sets,
so no prefix tree is needed to achieve constant access time and space.

We note without proof that there is a deterministic variant of the algorithm that uses space $O(n^2 m)$
and time $O(n^3 m)$.
This variant is more complicated, so we chose not to present it.


\paragraph{Extending the algorithm to other inequality tests.}
Our model considers decision trees that use less-than and equality tests.
Allowing the negations of these tests is a trivial extension.
(E.g., every greater-than-or-equal test $\test \ge k$
is equivalent by swapping the children to the less-than test $\test < k$.)
We note without proof that our results also extend easily to the model that allows less-than-or-equal tests
(of the form $\test \le k$).
Such tests only need to be accounted for
in the proof of Theorem~\ref{thm: intermediate};
the extended algorithm then allows such tests in Recurrence~\ref{recurrence}.



\paragraph{Acknowledgements.}
Thanks to Mordecai Golin and Ian Munro for introducing us to the problem and for useful discussions.



\bibliographystyle{ACM-Reference-Format}
\bibliography{bib}

\end{document}